\documentclass[letterpaper,10 pt, conference]{ieeeconf}
\IEEEoverridecommandlockouts                               
\usepackage{fancyhdr}
\usepackage{hyperref}
\usepackage{xspace}
\usepackage[font={small}]{caption}
\usepackage{color}
\usepackage{setspace}
\usepackage{enumerate}
\usepackage{graphics}
\usepackage{cite}
\usepackage{dsfont}
\usepackage{latexsym}
\usepackage{amsmath}
\usepackage{amssymb}
\usepackage{amsfonts}
\usepackage{amsthm}
\usepackage{times}
\usepackage{sgame}
\usepackage{color}
\usepackage{verbatim}
\usepackage{xcolor}
\usepackage{pgfplots}
\usepackage{epsfig} 
\usepackage{subcaption}
\usepackage{hyperref}

\usepackage{multirow,array}

\let\labelindent\relax
\usepackage{enumitem}
\newlist{inparaenum}{enumerate}{2}
\setlist[inparaenum]{nosep}
\setlist[inparaenum,1]{label=\bfseries\arabic*.}
\setlist[inparaenum,2]{label=\arabic{inparaenumi}\emph{\alph*})}

\def\mbb{\mathbb}

\def\mcc{\mathcal{C}}
\def\mcd{\mathcal{D}}

\def\dSP0{\delta_{SP0}}
\def\dRT0{\delta_{RT0}}
\def\dPS1{\delta_{PS1}}
\def\dTR1{\delta_{TR1}}
\def\half{\frac{1}{2}}
\def\prt{\partial}
\def\dgdx{\frac{\partial g}{\partial x}}
\def\dgdn{\frac{\partial g}{\partial n}}

\def\tr{\text{tr}}

\def\inv{{-1}}
\def\int{ {\text{int}} }
\def\toc{\text{t}}

\def\limb{\lim_{\beta\rightarrow\infty}}

\theoremstyle{plain}
\newtheorem{theorem}{Theorem}[section]

\newtheorem{proposition}{Proposition}[section]
\newtheorem{lemma}{Lemma}[section]

\newtheorem{assumption}{Assumption}

\def\bs{\boldsymbol}

\normalsize\title{\LARGE \bf
	The madness of people: rational learning in feedback-evolving games
	\thanks{ }}

\author{
	Keith Paarporn
	\thanks{K. Paarporn is with the Department of Computer Science, University of Colorado, Colorado Springs. Contact: \texttt{ \{kpaarpor\}@uccs.edu}.
	}
}

\begin{document}
\thispagestyle{plain}
\pagestyle{plain}

\maketitle

\begin{abstract}
 	The replicator equation in evolutionary game theory describes the change in a population's behaviors over time given suitable incentives. It arises when individuals make decisions using a simple learning process -- imitation. A recent emerging framework builds upon this standard model by incorporating game-environment feedback, in which the population's actions affect a shared environment, and in turn, the changing environment shapes incentives for future behaviors. In this paper, we investigate game-environment feedback when individuals instead use a boundedly rational learning rule known as logit learning. We characterize the resulting system's complete set of fixed points and their local stability properties, and how the level of rationality determines overall environmental outcomes in comparison to imitative learning rules. We identify a large parameter space for which logit learning exhibits a wide range of dynamics as the rationality parameter is increased from low to high. Notably, we identify a bifurcation point at which the system exhibits stable limit cycles. When the population is highly rational, the limit cycle collapses and a tragedy of the commons becomes stable.	
\end{abstract}

\section{Introduction}

The ``Tragedy of the Commons" refers to a scenario in which individuals acting according to their own self-interest leads to the destruction of a shared common resource \cite{Hardin_1968}. The originating example describes a group of cattle herders that share a common pasture land, which becomes overgrazed as each herder allows more of their cows to use it. Indeed, individual incentives are often mis-aligned with collective benefits that could be realized through mutual cooperation. It is relevant to many scenarios: there are economic and personal costs in reducing emissions, quarantining during a pandemic, and conserving resources such as water or gas \cite{ostrom1990governing}.

Game theory is a powerful tool that can predict population-level behaviors provided that individuals' incentives can be modeled. Classical  formulations predict outcomes when the  incentives are static, i.e. they do not change over time. However, actions have consequences on the environment, and a changing environment in turn affects individuals' incentives. For example, when the prevalence of an infectious disease is high, people will tend to stay at home as it becomes more likely to get infected. When the prevalence becomes lower, people will start to resume normal activities -- however, this can encourage the spread of new infections \cite{funk2010modelling,weitz2020awareness}. An emerging framework  termed "feedback-evolving games" incorporates a dynamic coupling between population-level behaviors and its impact on environmental states \cite{weitz2016oscillating}.

Feedback-evolving games constitutes a flexible framework capable of modeling the coupling between population behaviors and relevant environmental systems, such as social behaviors in epidemics, behaviors in climate change, and consumption of common resources \cite{satapathi2023coupled,khazaei2021disease,frieswijk2022modeling,tilman2020evolutionary}. Extensive research has characterized many possible dynamics that can emerge \cite{weitz2016oscillating,tilman2020evolutionary,gong2022limit,stella2022lower,paarporn2023sis}. These feedback-evolving models primarily consider population behaviors that are governed by the replicator dynamics (notable exceptions are \cite{stella2021impact,arefin2021imitation}). The replicator dynamic arises from  simple imitative learning rules at the individual level: an agent changes its action if it observes another agent that is more successful using a different action. The main assumptions underlying imitative learning is that agents do not utilize sophisticated cognitive abilities to make a decision \cite{sandholm2010population}. However, people sometimes make rational choices (i.e. payoff-maximizing), and sometimes make irrational ones (suboptimal) due to noise in their decision-making or available information \cite{blume2003noise}.

In this paper, we consider a feedback-evolving game where agents make boundedly rational choices. Instead of imitation, agents are logit learners. The logit rule is parameterized by a rationality parameter $\beta \geq 0$. When $\beta = 0$, agents blindly choose an action uniformly at random. As $\beta$ becomes higher, agents make a payoff-maximizing choice at higher rates, and a suboptimal choice at lower rates. In the limit of large $\beta$, the logit rule converges to a best-response. Logit learning is fundamentally different from imitation, as it requires agents to have access to information about payoffs from all strategies. The well-known algorithm called ``log-linear learning" in finite player settings has extensively been studied in regards to its convergence properties in potential games \cite{marden2012revisiting,tatarenko2014proving,blume1995statistical} and networked coordination games \cite{auletta2012metastability,paarporn2020impact,paarporn2020risk,zhang2023rationality}.

The primary contribution of this paper is the analysis of the dynamics induced by logit learning for varying levels of the rationality parameter. We focus our study on a parameter regime in which imitative learning is known to lead to a tragedy of the commons as the globally stable outcome-- all agents defect and the environment collapses. Thus, our study is also aimed at determining the effectiveness of rational learning in stabilizing more desirable environmental outcomes.  Interestingly, the logit system exhibits a variety of dynamics that range from a tragedy of the commons to limit cycles. A summary of our results is depicted in Figure \ref{fig:summary}.


We provide preliminary background on feedback-evolving games  in Section \ref{sec:prelim}. The proposed logit dynamics are presented in Section \ref{sec:logit_model}. Here, we identify the complete set of fixed points of this system and conditions for their stability (Theorem \ref{thm:logit_FPs}). We note that the set of fixed points differ from those in the original, imitative system. In Section \ref{sec:bif}, we more closely analyze properties of an interior fixed point. We identify the rationality level where it undergoes a Hopf bifurcation, which gives rise to stable limit cycles (Theorem \ref{thm:Hopf}).

\section{Background: feedback-evolving games}\label{sec:prelim}

\begin{figure}[t]
	\centering
	\includegraphics[scale=.33]{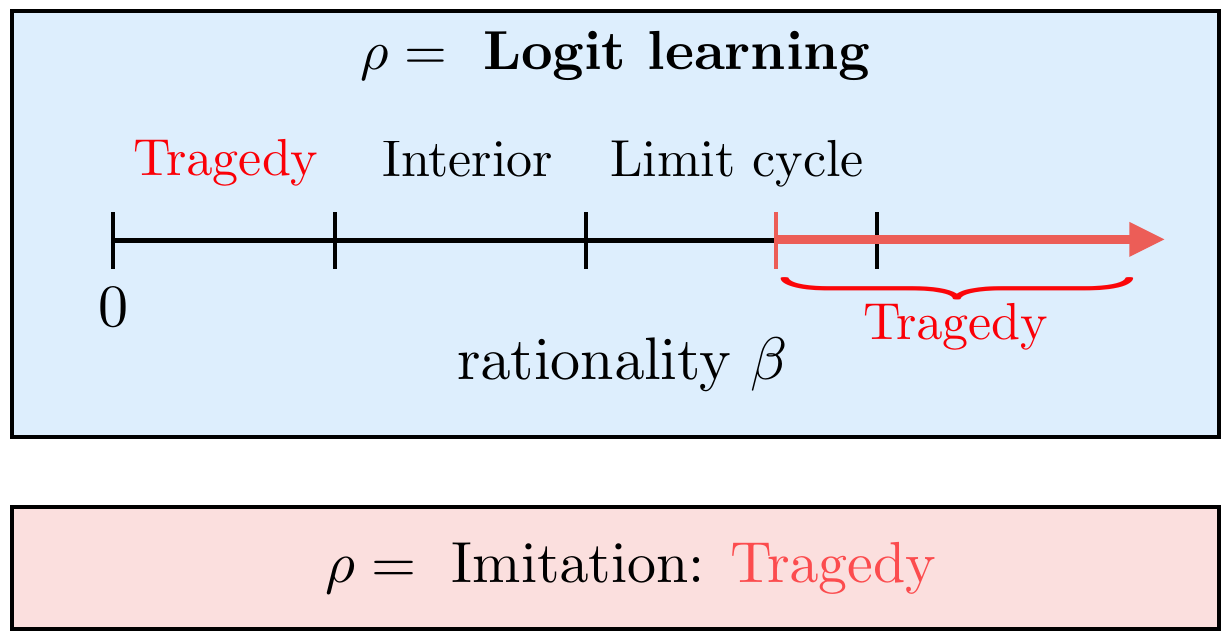}
	\caption{Contributions: summary of dynamics of logit learning in feedback-evolving games. In increasing order of the population's rationality $\beta$, the system exhibits: 1) a tragedy of the commons (TOC), 2) an interior fixed point (sustained non-zero resource), 3) a bifurcation into stable and growing limit cycles, and  4) again a TOC. The system also exhibits bistability between a tragedy and a  limit cycle. When agents instead follow imitative protocols, the dynamics always leads to a tragedy. }
	\label{fig:summary}
\end{figure}

A feedback-evolving game considers a population of agents whose actions have consequences on the abundance of an environmental state or shared resource, $n\in[0,1]$ (Figure \ref{fig:FEG}). At any given time, an agent chooses whether to cooperate ($\mcc$) or defect ($\mcd$). The defect action degrades $n$ (e.g. high resource consumption), and the cooperate action contributes to improving $n$ (e.g. restrained consumption). The immediate payoff available to each agent is dependent on the current environmental condition:
\begin{equation}
	A_n = n\begin{bmatrix} R_1 & S_1 \\ T_1 & P_1 \end{bmatrix} + (1-n)\begin{bmatrix} R_0 & S_0 \\ T_0 & P_0 \end{bmatrix}
\end{equation}
The $2\times 2$ payoff matrix describes the immediate rewards that are available to the agents, where the first strategy corresponds to an agent adopting $\mcc$, and the second strategy corresponds to an agent adopting $\mcd$. Denoting $x \in [0,1]$ as the fraction of cooperating agents in the population, the reward to a cooperating and defecting agent is given by 
\begin{equation}
	\pi_\mcc(x,n) = [A_n [x,1-x]^\top ]_1, \quad \pi_\mcd(x,n) = [A_n [x,1-x]^\top ]_2
\end{equation}
respectively. We will denote the payoff difference between cooperation and defection as: 
\begin{equation}
	\begin{aligned}
		g(x,n) &:= \pi_\mcc(x,n) - \pi_\mcd(x,n) \\
		&= a xn + bx + cn + d
	\end{aligned}
\end{equation}
where we denote
\begin{equation}
	\begin{aligned}
		a &:= \dSP0 - \dRT0 + \dPS1 - \dTR1 \\
		b &:= \dRT0 - \dSP0 \\
		c &:= -(\dPS1 + \dSP0) \\
		d &:= \dSP0.
	\end{aligned}
\end{equation}
and $\dTR1 = T_1-R_1$, $\dPS1 = P_1 - S_1$, $\dRT0 = R_0 - T_0$, and $\dSP0 = S_0 - P_0$. The $\delta$ constants are referred to as payoff parameters. 

\begin{assumption}\label{assume:A1}
	Defection is the dominant strategy in the $2\times 2$ game that corresponds to the payoff matrix $A_1$. In particular, $\dTR1  > 0$ and $\dPS1 > 0$.
\end{assumption}
The above assumption is widely adopted in the feedback-evolving games literature. It asserts that agents have more incentives to consume resources when they are abundant ($n=1$).

\begin{figure}[t]
	\centering
	\includegraphics[scale=.33]{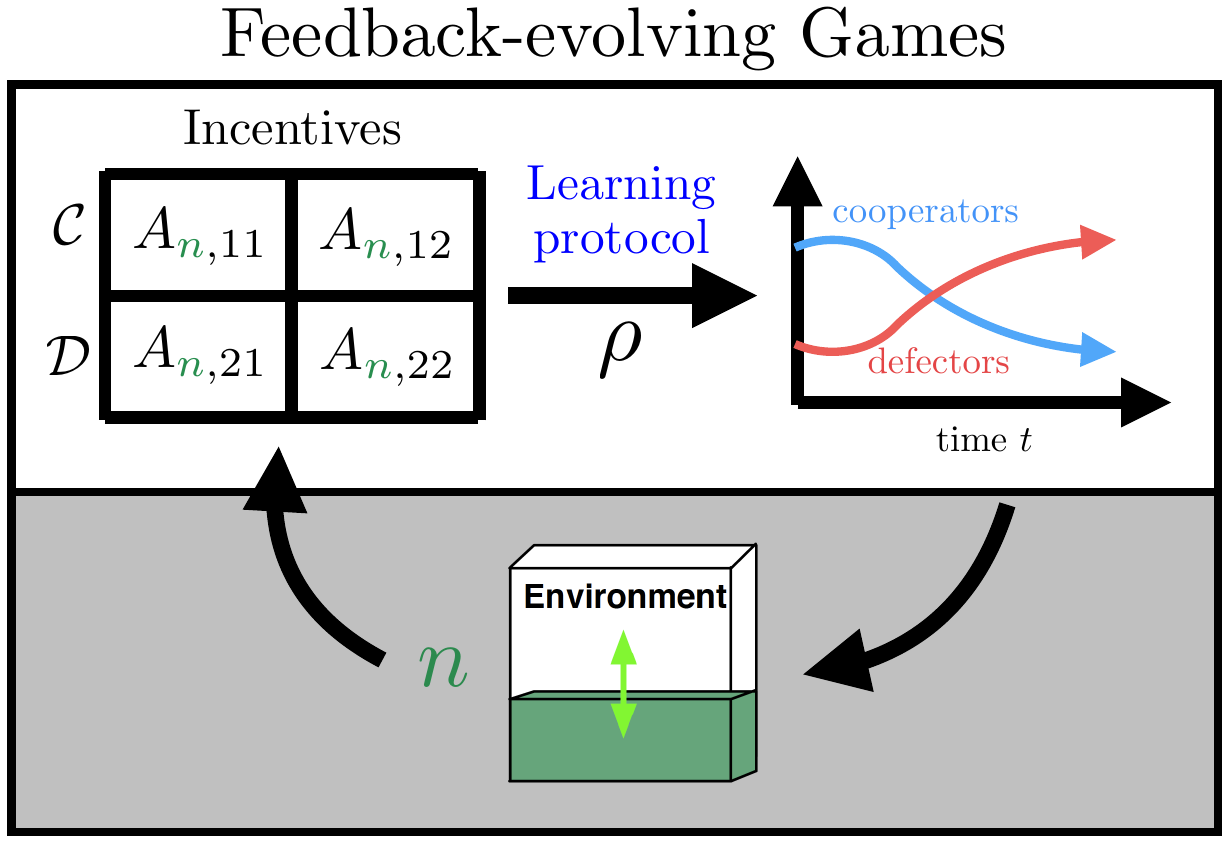}
	\caption{The feedback-evolving games framework. Individuals' actions affect a shared environment, which shapes incentives for future actions. The agents' decision-making process is specified by the learning protocol $\rho$. }
	\label{fig:FEG}
\end{figure}

The agents dynamically update their decisions over time. In the standard analyses of feedback-evolving games, agents are assumed to follow a revision protocol that induces the replicator dynamics. A revision protocol is a description of the behavioral dynamics of agents in the population. It is specified by a function $\rho_{ij}(x,n)$ that gives the rate at which an agent currently adopting strategy $i$ switches to strategy $j$. For a given revision protocol $\rho_{ij}$, the \emph{mean population dynamics} describing the change in cooperator fraction over time is generically given by the rate equation
\begin{equation}\label{eq:mean_dynamics}
	\dot{x} = (1-x)\rho_{\mcd\mcc}(x,n) - x\rho_{\mcc\mcd}(x,n).
\end{equation}
Imitative revision protocols induce the replicator dynamics -- one such example is the imitative pairwise comparison protocol
\begin{equation}
	\rho_{\mcd\mcc}(x,n) = x [g(x,n)]_+, \quad \rho_{\mcc\mcd}(x,n) = (1-x) [-g(x,n)]_+
\end{equation}
where $[a]_+ = \max\{0,a\}$. A defecting agent will switch to cooperate only if $g(x,n) > 0$, and a cooperating agent will switch to defection only if $g(x,n) < 0$. The overall coupled game-environment system dynamics considered in \cite{weitz2016oscillating} is then given by
\begin{equation}\label{eq:standard}
	\begin{aligned}
		\dot{x} &= x(1-x)g(x,n) \\
		\dot{n} &= \epsilon n(1-n)(\theta x - (1-x))
	\end{aligned}\tag{ID}
\end{equation}
The form of the environmental dynamics $F_n(x,n)$ is referred to as the \emph{tipping point dynamics}, and has been extensively studied in the literature \cite{gong2022limit,stella2021impact,stella2022lower,tilman2020evolutionary,weitz2016oscillating}. The environment does not improve unless a sufficient fraction $(1+\theta)^\inv$ of the population cooperates. The parameter $\epsilon > 0$ is a time-scale separation constant. The cooperators help restore the environment at the rate $\theta$, and defectors degrade $n$ at a unit rate. The state $\bs{z} = (x,n)$ evolves over the state space $\Gamma := [0,1]^2$. By inspection, one can verify that $\Gamma$ is forward-invariant with respect to the dynamics \eqref{eq:standard}. We will classify two types of fixed points of the feedback-evolving system \eqref{eq:standard}. A \emph{tragedy of the commons (TOC)} is a fixed point of the form $\bs{z}_\toc = (x_\toc,0)$. Such an outcome indicates that the environmental resource has totally collapsed. A \emph{prosperity} fixed point is of the form $(x,1)$. An \emph{interior fixed point} is one such that $\bs{z}^* \in \text{int}(\Gamma)$, i.e. $x^*,n^* \in (0,1)$. The imitative system \eqref{eq:standard} has four corner fixed points $(0,0), (0,1), (1,0), (1,1)$, and under some parameter regimes, a unique interior fixed point.

The goal of this paper is to characterize how the above system dynamics qualitatively differ when agents follow an alternate revision protocol known as \emph{logit learning}, which endows agents with some degree of rationality. Our comparative analysis will focus on a parameter regime in which the standard imitative dynamics \eqref{eq:standard} leads to a tragedy of the commons. First, we will assume
\begin{assumption}\label{assume:theta}
	The replenishment rate $\theta < 1$. 
\end{assumption} 
In words, defection degrades the resource faster than cooperation restores it. Consequently, an irrational population will cause a tragedy of the commons. Moreover, we consider the following condition on payoff parameters in the collapsed state:
\begin{assumption}\label{assume:A0}
	We will consider payoff parameters $\dSP0 < 0$ and $\dRT0 > - \dSP0$.
\end{assumption}

Under Assumptions \ref{assume:theta} and \ref{assume:A0}, it is established in \cite{weitz2016oscillating} that the tragedy fixed point $(0,0)$ is globally attracting, and there exists a unique and unstable interior fixed point
\begin{equation}\label{eq:nbar}
	x_\int := \frac{1}{1+\theta}, \quad \bar{n} := \frac{\dRT0 + \theta \dSP0}{\dRT0 + \dTR1 + \theta(\dSP0 + \dPS1)}
\end{equation}


\section{Model: logit learning}\label{sec:logit_model}

Suppose agents follow a perturbed best-response dynamic called the logit protocol \cite{sandholm2010population}. This is a departure from usual considerations that the agents are imitative learners. The logit revision protocol is given by
\begin{equation}
	\begin{aligned}
	\rho_\mcc(x,n) &=  \frac{e^{\beta \pi_\mcc(x,n)}}{ e^{\beta \pi_\mcc(x,n)} + e^{\beta \pi_\mcd(x,n)} } \\ \rho_\mcd(x,n) &=  \frac{e^{\beta \pi_\mcd(x,n)}}{ e^{\beta \pi_\mcc(x,n)} + e^{\beta \pi_\mcd(x,n)} }
	\end{aligned}
\end{equation}
where $0 \leq \beta < \infty$ is the rationality parameter of an agent. The logit protocol is fundamentally different from imitative protocols. The protocol $\rho_\mcc$ (resp. $\rho_\mcd$) describes the switch rate to strategy $\mcc$ ($\mcd$) for any agent in the population.  For low values of $\beta$, agents choose their actions uniformly at random, and for high values of $\beta$, they select the payoff-maximizing action with a probability close to 1. From \eqref{eq:mean_dynamics}, the logit protocol induces the mean dynamics
\begin{equation}
	\begin{aligned}
		\dot{x} &= (1-x)\rho_{\mcc}(x,n) - x\rho_{\mcd}(x,n) \\
		&= (1-x)\rho_{\mcc}(x,n) - x(1-\rho_{\mcc}(x,n)) \\
		&= \rho_\mcc(x,n) - x
	\end{aligned}
\end{equation}
and overall, the coupled game-environment system dynamics are:
\begin{equation}\label{eq:logit_system}
	\boxed{
	\begin{aligned}
		\dot{x} &= F_1(x,n) := \frac{e^{\beta g(x,n)}}{1 + e^{\beta g(x,n)}} - x \\
		\dot{n} &= F_2(x,n) := n(1-n)(\theta x - (1-x))
	\end{aligned}\tag{LD}}
\end{equation}
The  system \eqref{eq:logit_system} will be the main focus of this paper. One can verify that the state space $\Gamma$ is forward-invariant through an application of Nagumo's Theorem: whenever $x = 0$ or $x=1$, $\dot{x} > 0$ or $\dot{x} < 0$, respectively. Moreover, when $n=0$ or $n=1$, it holds that $\dot{n} = 0$. 

\subsection{Characterization of fixed points}

We will classify a \emph{logit interior fixed point} as a fixed point of system \eqref{eq:logit_system} that satisfies $\bs{z}_l^* = (x^*,n^*) \in \text{int}(\Gamma)$. We observe that when $\beta = 0$, every agent chooses an action uniformly at random, and thus in equilibrium, $x = 1/2$. The equilibrium environmental state is then determined by the value of $\theta$: under the assumption $\theta < 1$, a completely irrational population  causes a tragedy $n=0$.

It is important to note that the logit system does not share any of the  fixed points as the imitative system \eqref{eq:standard}. However, we may still classify fixed points as either TOC or interior. The complete set of fixed points of \eqref{eq:logit_system} is characterized in the result below.
\begin{theorem}\label{thm:logit_FPs}
	The fixed points of system \eqref{eq:logit_system} are characterized as follows. 
	\begin{enumerate}
	\item Suppose $\beta = 0$.  If $\theta \neq 1$, then there are two fixed points $({\half},0)$ and $({\half},1)$. If $\theta = 1$, then a line of equilibria $({\half},n)$ for all $n \in [0,1]$ exists.
	\item Suppose $\beta > 0$. A unique  interior fixed point $\bs{z}_\int^* = (x_\int,n_\int)$ exists if and only if
	\begin{equation}\label{eq:logit_interior_condition}
		\beta \geq \beta_\int := \frac{(1+\theta)\log \theta^\inv}{\dRT0 + \theta\dSP0}
	\end{equation}
	where $x_\int$ was given in \eqref{eq:nbar} and $n_\int$ is given by
	\begin{equation}\label{eq:logit_FP}
		n_\int = \frac{\dRT0 + \theta \dSP0 - \beta^\inv(1+\theta) \log \frac{1}{\theta} }{\dRT0 + \theta\dSP0 + \dTR1 + \theta\dPS1} \in (0,1)
	\end{equation}
	
	If $\beta < \beta_\int$, no interior fixed points exist. 
	\item Suppose $\beta > 0$. There exists a $\hat{\beta} > 4/b$ such that:
	\begin{enumerate}
		\item For all $\beta\in (0,\hat{\beta})$, there is a unique TOC fixed point $x_{\toc3}(\beta)$. It holds that it is strictly increasing in $\beta$, $x_{\toc3}(\beta) > {\half}$, $x_{\toc3}(\beta_\int) = x_\int$, and $\limb x_{\toc3}(\beta) = 1$.
		\item For all $\beta \geq \hat{\beta}$, there are three TOC fixed points $x_{\toc1}\leq x_{\toc2} < {\half} < x_{\toc3}$, where the equality holds if and only if $\beta = \hat{\beta}$. It holds that $x_{\toc1}$ is strictly decreasing with $\limb x_{\toc1}(\beta) = 0$, and $x_{\toc2}$ is strictly increasing with $\limb x_{\toc2}(\beta) = \frac{|\dSP0|}{\dRT0}$.
	\end{enumerate}


	\item Suppose $\beta > 0$. There exists a unique prosperity fixed point of the form $(x^*,1)$, where $x^* < {\half}$ is strictly decreasing in $\beta$ and $\limb x^*(\beta) = 0$.
	\end{enumerate}
\end{theorem}
The environmental level at the interior fixed point $n_\int$ is monotonically increasing in the rationality $\beta > \beta_\int$. At the threshold $\beta = \beta_\int$, $n_\int = 0$, and as $\beta \rightarrow \infty$, the level $n_\int$ approaches $\bar{n}$, the interior fixed point from the imitative system \eqref{eq:nbar}. For identifying TOC fixed points, the fixed point equation \eqref{eq:logit_system} is transcendental, and thus its solutions cannot be expressed generally in closed form. Consequently, the precise value of $\hat{\beta}$ in item 3) cannot be analytically derived.

\subsection{Stability properties of fixed points}

To conclude this section, we establish the stability properties of all fixed points except the interior FP, which we will closely investigate in the next section. When $\beta=0$, the only two fixed points are $(1/2,0)$ and $(1/2,1)$. From Assumption \ref{assume:theta}, we immediately deduce that $(1/2,0)$ is stable and $(1/2,1)$ is unstable. Since no interior fixed point exists at $\beta = 0$, we may also conclude that $(1/2,0)$ is globally attractive by invoking Poincare-Bendixson Theorem (there cannot be any orbits in $\int \ \Gamma$). It will be useful to derive the entries of the Jacobian $J(x,n)$ of \eqref{eq:logit_system}. They are:
\begin{equation}\label{eq:Jacobian}
    		\begin{aligned}
    			J_{11} = \frac{\partial F_1}{\partial x} &= \frac{\beta \dgdx(n) e^{\beta g(x,n)} }{(1+e^{\beta g(x,n)})^2} - 1 \\
    			J_{12} = \frac{\partial F_1}{\partial n} &= \frac{\beta \dgdn(x) e^{\beta g(x,n)} }{(1+e^{\beta g(x,n)})^2}  \\
    			J_{21} = \frac{\partial F_2}{\partial x} &= \epsilon n(1-n)(1+\theta) \\
    			J_{22} = \frac{\partial F_2}{\partial n} &= \epsilon (1-2n)(\theta x - (1-x))
    		\end{aligned}
\end{equation}

The following result details the stability properties of TOC and prosperity fixed points when $\beta > 0$.

\begin{proposition}\label{prop:FP_stability}
	Suppose $\beta > 0$. 
	\begin{enumerate}
	\item A TOC fixed point $(x_\toc,0)$ is locally stable if and only if
	\begin{equation}\label{eq:TOC_stable}
		x_\toc < x_\int \text{ and } \beta b x_\toc(1-x_\toc) < 1.
	\end{equation}
	For $\beta \geq \beta_\int$, $x_{\toc3}$ is unstable. For sufficiently large $\beta \geq \hat{\beta}$, $x_{\toc1}$ is stable and $x_{\toc2}$ is unstable.
	\item The fixed point of the form $(x^*,1)$ is unstable.
	\end{enumerate}
\end{proposition}
This result implies that only TOC and interior fixed points can be  stable.
\begin{proof}
	Let us first consider any TOC fixed point $(x_\toc,0)$. The Jacobian  is
	\begin{equation}
	\begin{bmatrix} \beta b x_\toc(1-x_\toc) - 1 & \beta(a x_\toc + c) x_\toc(1-x_\toc) \\ 0 & \epsilon( (1+\theta)x_\toc - 1 ) 	\end{bmatrix}
	\end{equation}
	Since this is an upper triangular matrix, the eigenvalues are its diagonal entries. The fixed point is thus stable under the condition \eqref{eq:TOC_stable}. By Theorem \ref{thm:logit_FPs}, $x_{\toc3}$ is unstable for all $\beta \geq \beta_\int$ since $x_{\toc3}(\beta) \geq x_\int$. Now, let us consider the other two TOC fixed points $x_{\toc1}, x_{\toc2}$ for $\beta \geq \hat{\beta}$ (Item 3 of Theorem \ref{thm:logit_FPs}). The first eigenvalue for both fixed points are negative, since  $x_{\toc1}, x_{\toc2} < \half < x_\int$, and so their stability rests on the sign of the second eigenvalue. Any TOC fixed point solves the equation
	\begin{equation}
		(1-x_\toc)^2 e^{\beta (b x_\toc + d)}  = x_\toc (1-x_\toc)
	\end{equation}
	which simply follows from the equilibrium condition $F_1(x_\toc,0) = 0$. Then the sign of the second eigenvalue is negative if 
	\begin{equation}
		b \beta (1-x_\toc)^2 e^{\beta (b x_\toc + d)} < 1
	\end{equation}
	Focusing on $x_{\toc1}$, we recall that $\limb x_{\toc1} = 0$. For large $\beta$, we have
	\begin{equation}\label{eq:xtoc1_limit}
		\begin{aligned}
			&\limb b \beta (1-x_{\toc1})^2 e^{\beta (b x_{\toc1} + d)} \\ 
			&= b \left(\limb \beta (1-x_{\toc1}(\beta))^2 \right) \cdot \left( \limb e^{\beta (b x_{\toc1}(\beta) + d)}\right) \\
		\end{aligned}
	\end{equation}
	The first limit above can be written as the product of limits $(\limb \beta)\cdot \limb (1-x_{\toc1}(\beta))^2 = \limb \beta$. The second limit above can be written as
	\begin{equation}
		\begin{aligned}
		&e^{(\limb \beta)\cdot (\limb (bx_{\toc1} + d) )} \\
		&= e^{(\limb \beta) d} \\
		&= \limb e^{\beta d}
		\end{aligned}
	\end{equation}
	We can thus re-express \eqref{eq:xtoc1_limit} as
	\begin{equation}
		\begin{aligned}
			b (\limb \beta)\cdot (\limb e^{\beta d }) = b \limb \beta e^{\beta d} = 0.
		\end{aligned}
	\end{equation}
	where the last equality follows since $d < 0$ (Assumption \ref{assume:A0}) and from the fact that $e^{\beta d}$ is a negiglible function (Ch. 3 \cite{katz2014introduction}). Therefore, there exists a $\beta_1 \geq \hat{\beta}$ such that for all $\beta \geq \beta_1$, the fixed point $(x_{\toc1},0)$ is locally stable.
	
	Now, we consider $x_{\toc2}$. Recall $\limb x_{\toc2} = \frac{|\dSP0|}{\dRT0}$. For large $\beta$,
	\begin{equation}
		\begin{aligned}
	 		\limb b\beta x_{\toc2} (1- x_{\toc2}) &= b \left( \limb \beta \right)  \left( \limb x_{\toc2} (1- x_{\toc2})  \right) \\
			&= b\left( \limb \beta \right) \frac{|\dSP0|}{\dRT0} (1 - \frac{|\dSP0|}{\dRT0}) \\
			& = +\infty
		\end{aligned}
	\end{equation}
	since $b > 0$. Thus, there exists a $\beta_2 \geq \hat{\beta}$ such that for all $\beta \geq \beta_2$, the  fixed point $(x_{\toc2},0)$ is unstable.
	
	Now, we consider the fixed point $(x^*,1)$. The Jacobian evaluated here is
	\begin{equation}
		 \begin{bmatrix} \beta (a+b) x(1-x) - 1 & \beta(a x + c) x(1-x) \\ 0 & \epsilon (1-(1+\theta)x)  	\end{bmatrix}
	\end{equation}
	The second eigenvalue is positive if and only if $x^* <  x_\int$. It was established in item 4 of Theorem \ref{thm:logit_FPs} that $x^* < \half < x_\int$. Consequently, the fixed point cannot be stable for any $\beta$.
\end{proof}
\section{Bifurcations from logit learning}\label{sec:bif}

In this section, we take the rationality level $\beta > 0$ as a bifurcation parameter of the logit system \eqref{eq:logit_system}. We study the stability properties of the interior fixed point as $\beta$ increases. Notably, we establish a critical value ${\beta_h}$ at which it undergoes a Hopf bifurcation. That is, for a neighborhood of values $\beta > {\beta_h}$, the system exhibits a stable limit cycle around the fixed point $(x_\int,n_\int)$ whose amplitude grows in $\beta$. 

\subsection{Bifurcation of limit cycles}

\begin{figure}
	\centering
	\includegraphics[scale=0.33]{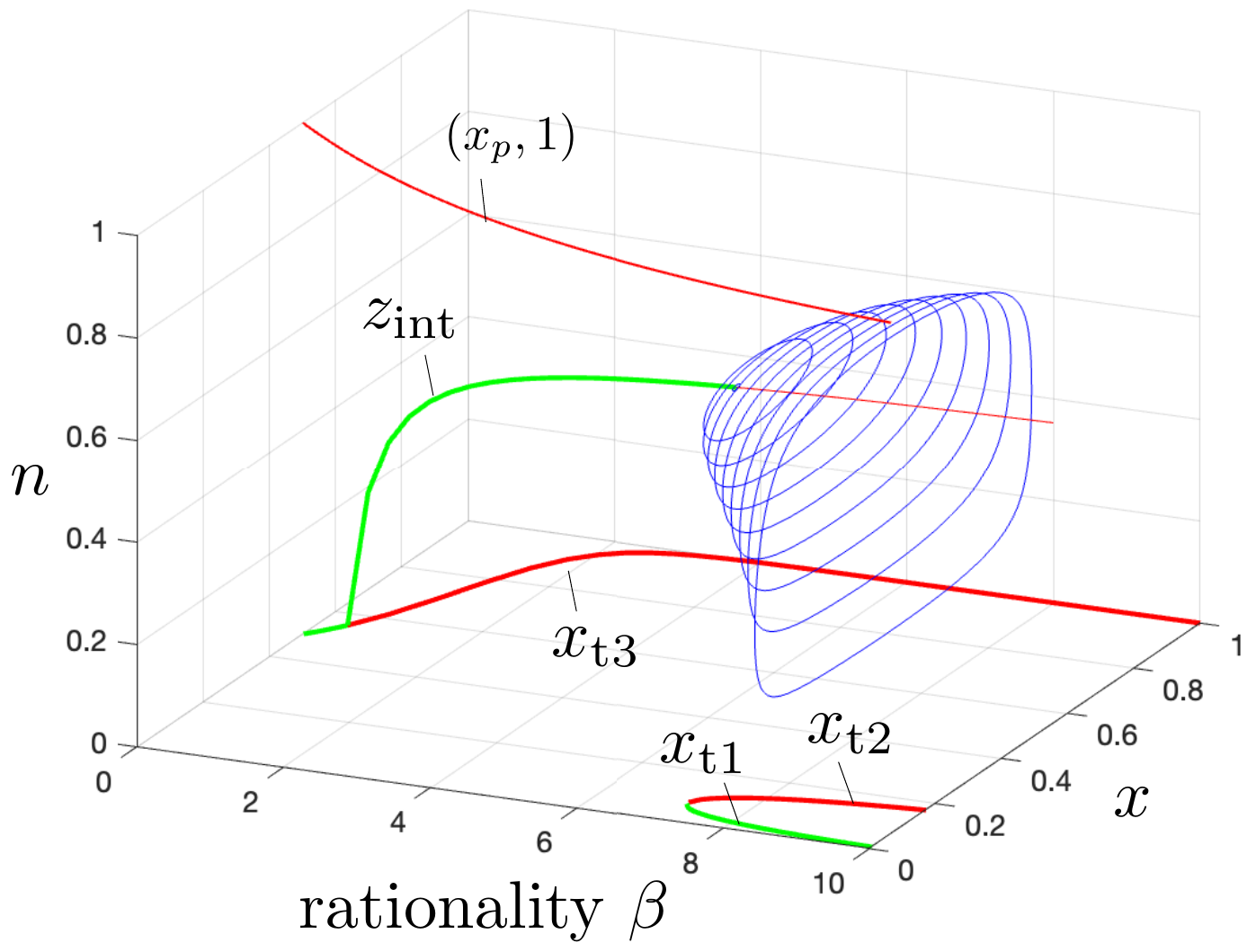}
	\caption{Bifurcation plot of the logit system \eqref{eq:logit_system}. Fixed parameters are: $\epsilon = 0.5$, $\dTR1 = 0.5$, $\dPS1 = 0.25$, $\dSP0 = -0.5$, $\dRT0 = 1.5$, $\theta = 0.8$. Thresholds are: $\beta_\int = 0.3651$, ${\beta_h} = 5.6767$. The green lines indicate stable, isolated fixed points. The red lines indicate unstable isolated fixed points. The blue lines are stable limit cycles. They exist only in the interval $\beta \in ({\beta_h},\beta_u)$, where $\beta_u \approx 7.84$. For $\beta > \beta_u$, The TOC fixed point $x_{\toc1}$ appears globally stable.}
	\label{fig:bifurcation}
\end{figure}

\begin{figure*}
	\centering
	\includegraphics[scale=0.25]{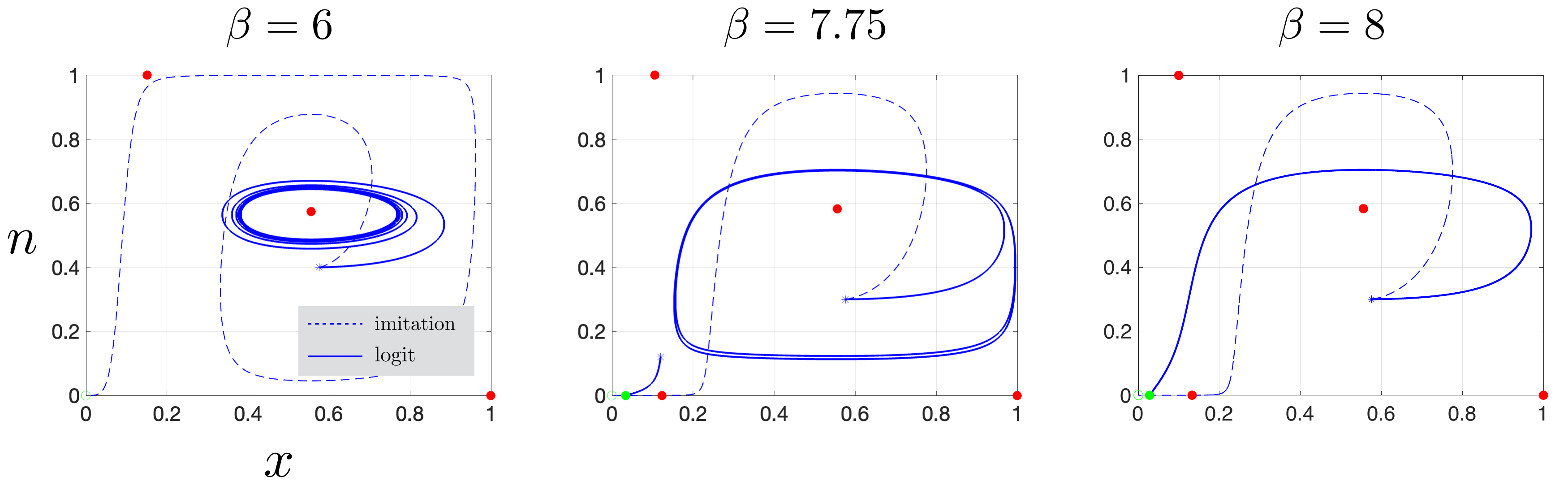}
	\caption{Trajectories of the system dynamics. The solid blue lines correspond to the logit dynamics \eqref{eq:logit_system}. The dashed blue lines correspond to the standard imitative dynamics \eqref{eq:standard}. The blue asterisk is the initial condition. The filled circles are fixed points of the logit system, where red indicates it is unstable and green indicates it is locally stable. The open green circle at $(0,0)$ is the TOC fixed point of the imitative dynamics. Here, we use the same parameter setup as in Figure \ref{fig:bifurcation}, where the bifurcation value is $\beta_h = 5.67$. (Left) We see the trajectory converges to a stable limit cycle when $\beta = 6$. (Center) The system exhibits bistability when $\beta = 7.75$ -- some initial conditions are attracted to $x_{\toc1}$, while others are attracted to a stable limit cycle. (Right) For $\beta = 8$, the limit cycle has disappeared and the only stable fixed point of the system is $x_{\toc1}$.}
	\label{fig:portrait}
\end{figure*}

We first state the Hopf bifurcation theorem below.
\begin{theorem}[Hopf Bifurcation Theorem (Ch. 3 \cite{guckenheimer2013nonlinear})]\label{thm:Hopf}
	Consider a dynamical system $\dot{z} = F(z;\beta)$, where $z \in \mbb{R}^2$ is the state and $\beta \in \mbb{R}$ is a bifurcation parameter. Suppose the system has an equilibrium $(z^*;\beta^*)$ (where $z^*$ may depend on $\beta$) at which the following properties hold:
	\begin{enumerate}
		\item The Jacobian evaluated at $(z^*;\beta^*)$ has a pair of pure imaginary eigenvlaues $\lambda(\beta^*) = \pm\omega i$.
		\item $\frac{d\text{Re}(\lambda(\beta))}{d\beta} \bigr\rvert_{\beta^*} \neq 0$.
	\end{enumerate}
	Then the dynamics undergo a Hopf bifurcation at $(z^*;\beta^*)$, which induces a family of periodic solutions in a sufficiently small neighborhood of $(z^*;\beta^*)$.
\end{theorem}

The properties of the interior fixed point  \eqref{eq:logit_FP} is summarized in the following result.

\begin{theorem}
	The interior fixed point $\bs{z}_\int(\beta) = (x_\int(\beta), n_\int(\beta))$  is locally stable for $\beta \in [\beta_\int,{\beta_h})$, where
	\begin{equation}
		{\beta_h} := (a\bar{n} + b)^\inv\left( \theta(1+\theta^\inv)^2 + \frac{a(1+\theta)}{D}\log \theta^\inv \right).
	\end{equation}
	At the value $\beta = {\beta_h}$, it undergoes a Hopf bifurcation where its eigenvalues are purely imaginary, and it becomes an unstable focus in a vicinity $\beta > {\beta_h}$.
\end{theorem}
\begin{proof}
	The Jacobian evaluated at the interior fixed point is
	\begin{equation}
		J_\int := \begin{bmatrix} \frac{\beta\dgdx^*}{\theta(1+\theta^\inv)^2} - 1 & \frac{\beta\dgdn^*}{\theta(1+\theta^\inv)^2}  \\ \epsilon n_\int(1-n_\int)(1+\theta) & 0 	\end{bmatrix}
	\end{equation}
	where for compactness, we write $\dgdx^*$ and $\dgdn^*$ to represent the partial derivatives evaluated at $\bs{z}_\int(\beta)$. The trace is 
	\begin{equation}
		\tr J_\int = \frac{\beta\dgdx^*}{\theta(1+\theta^\inv)^2}  - 1
	\end{equation}
	and the determinant is
	\begin{equation}
		\det J_\int =  - \epsilon \beta \dgdn^* n_\int(1-n_\int)\frac{1+\theta}{\theta(1+\theta^\inv)^2}.
	\end{equation}
	The fixed point is stable if the real parts of its eigenvalues are negative, which is equivalent to the condition that $\tr J_\int < 0$ and $\det J_\int > 0$. We have
	\begin{equation}
		\begin{aligned}
			\tr J_\int < 0 &\iff \beta(a n_\int + b) < \theta(1+\theta^\inv)^2 \\
			&\iff \beta < {\beta_h}
		\end{aligned}
	\end{equation}
	where the second line follows by re-writing $n_\int = \bar{n} - \frac{1+\theta}{D}\log \theta^\inv$ with $D := \dRT0 + \dTR1 + \theta(\dSP0 + \dPS1) > 0$, and observing that $a\bar{n} + b = \frac{1+\theta}{D}(\dRT0\dPS1 - \dSP0\dTR1) > 0$.
	Additionally, 
	\begin{equation}
		\begin{aligned}
			\det J_\int > 0 &\iff -\dgdn^* = -(ax_\int + c) > 0 \\ 
			&\iff -c > a(1+\theta)^\inv \\
			&\iff D > 0.
		\end{aligned}
	\end{equation}
	The sign of $\det J_\int$ is determined only from the payoff parameters, and does not depend on $\beta$. Thus, $\det J_\int > 0$ follows  from Assumption \ref{assume:A0}. This establishes the range of $\beta$ for which $\bs{z}_\int$ is stable. Its eigenvalues are given by
	\begin{equation}
		\begin{aligned}
			\lambda_1(\beta) &:= \frac{\tr J_\int}{2} + \sqrt{\left(\frac{\tr J_\int}{2}\right)^2 - \det J_\int} \\
			\lambda_2(\beta) &:= \frac{\tr J_\int}{2} - \sqrt{\left(\frac{\tr J_\int}{2}\right)^2 - \det J_\int} \\
		\end{aligned}
	\end{equation}
	At the bifurcation point $\beta = {\beta_h}$,  $J_\int$ has a conjugate pair of purely imaginary eigenvalues $\pm i \omega$ with $\omega = \sqrt{\det J_\int} > 0$. Moreover, the rate of change of the eigenvalues' real part is
	\begin{equation}
		\half\frac{\prt (\tr J_\int)}{\prt \beta} = \frac{a\bar{n}+b}{2\theta(1+\theta^\inv)^2} > 0.
	\end{equation}
	Indeed, $\tr J_\int$ is linearly increasing in $\beta$. Therefore, the interior fixed point is an unstable focus (positive real and non-zero imaginary parts) for all values $\beta > {\beta_h}$ that satisfy $\left(\frac{\tr J_\int}{2}\right)^2 < \det J_\int$.
\end{proof}
The Hopf bifurcation at ${\beta_h}$ asserts that a family of periodic cycles are guaranteed to appear for a  neighborhood of values $\beta > {\beta_h}$. Whether these periodic cycles are stable depends on the sign of the first Lyapunov coefficient $\ell_1$ evaluated at $(\bs{z}_\int;{\beta_h})$ (Ch. 3 \cite{guckenheimer2013nonlinear}). If $\ell_1 < 0$, then the bifurcated cycles are stable. The derivation of stability conditions for these cycles will be left for future work. However, we observe through extensive simulations that the bifurcated cycles are stable under the assumed parameter values.

So far, we have established that for $\beta \in [0,\beta_\int)$, the TOC fixed point $x_{\toc3}$ is the only stable fixed point (Proposition \ref{prop:FP_stability}). For $\beta \in [\beta_\int,{\beta_h})$, the interior fixed point $\bs{z}_\int$ is the only stable fixed point (Proposition \ref{prop:FP_stability} and Theorem \ref{thm:Hopf}). At $\beta={\beta_h}$, it bifurcates into an unstable focus and for a vicinity of values $\beta \geq {\beta_h}$, a limit cycle encircles $\bs{z}_\int$ (Theorem \ref{thm:Hopf}). A full bifurcation diagram that summarizes these findings is provided in Figure \ref{fig:bifurcation}.

\subsection{Simulations: the high rationality regime}

Under the parameter regime specified by Assumptions \ref{assume:A1}, \ref{assume:A0}, and \ref{assume:theta}, numerical simulations suggest there is another critical value $\beta_u > {\beta_h}$ for which the limit cycle collapses, and the TOC fixed point $x_{\toc1}$ becomes globally attractive. Indeed, Proposition \ref{prop:FP_stability} has established that $x_{\toc1}$ is stable for sufficiently high $\beta$, and it is the only stable fixed point in the system.

Simulations of system trajectories in the phase space are depicted in Figure \ref{fig:portrait}. In particular, we note that the system can exhibit bistability (center portrait) between the limit cycle and $x_{\toc1}$: initial conditions close to $(x_{\toc1},0)$ will converge to the TOC, and other conditions will converge to the stable limit cycle. Since the interior fixed point does not disappear for high $\beta$, we conjecture that the limit cycle collapses at some value $\beta_u$ for which its $\omega$-limit set touches the basin of attraction of $(x_{\toc1},0)$. Such an analysis is left for future work.

\section{Conclusion and Future Work}
In this paper, we formulated a feedback-evolving system where agents in a population follow a logit revision protocol, which is a boundedly rational learning rule. We analyzed the resulting dynamical outcomes as a function of the rationality parameter $\beta \geq 0$. In increasing order of $\beta$, we identified interval ranges for which the system exhibits 1) a tragedy of the commons (low rationality), 2) a stable interior fixed point, 3) stable limit cycles, and 4) again, a tragedy of the commons (high rationality). Counter-intuitively, high rationality leads to a collapsed environment, whereas moderate levels of rationality can lead to sustainable outcomes. Our analysis of the logit system holds in a parameter regime  where imitative learning leads to a tragedy of the commons. These results demonstrate that boundedly rational behaviors can induce a wide variety of environmental outcomes.

Future work will involve analyzing global stability properties of the system. Additionally, a complete analysis of the Lyapunov coefficient is needed to establish stability of the observed limit cycles, and the precise value of $\beta$ at which the limit cycle dissipates into the TOC outcome is yet to be established. The application of control strategies, e.g. incentivization of cooperation, to control global outcomes will also be studied.


\bibliographystyle{IEEEtran}
\bibliography{library}

\begin{thebibliography}{10}
\providecommand{\url}[1]{#1}
\csname url@samestyle\endcsname
\providecommand{\newblock}{\relax}
\providecommand{\bibinfo}[2]{#2}
\providecommand{\BIBentrySTDinterwordspacing}{\spaceskip=0pt\relax}
\providecommand{\BIBentryALTinterwordstretchfactor}{4}
\providecommand{\BIBentryALTinterwordspacing}{\spaceskip=\fontdimen2\font plus
\BIBentryALTinterwordstretchfactor\fontdimen3\font minus
  \fontdimen4\font\relax}
\providecommand{\BIBforeignlanguage}[2]{{%
\expandafter\ifx\csname l@#1\endcsname\relax
\typeout{** WARNING: IEEEtran.bst: No hyphenation pattern has been}%
\typeout{** loaded for the language `#1'. Using the pattern for}%
\typeout{** the default language instead.}%
\else
\language=\csname l@#1\endcsname
\fi
#2}}
\providecommand{\BIBdecl}{\relax}
\BIBdecl

\bibitem{Hardin_1968}
G.~Hardin, ``The tragedy of the commons,'' \emph{Science}, vol. 162, no. 3859,
  pp. 1243--1248, 1968.

\bibitem{ostrom1990governing}
E.~Ostrom, \emph{Governing the commons: The evolution of institutions for
  collective action}.\hskip 1em plus 0.5em minus 0.4em\relax Cambridge
  university press, 1990.

\bibitem{funk2010modelling}
S.~Funk, M.~Salath{\'e}, and V.~A. Jansen, ``Modelling the influence of human
  behaviour on the spread of infectious diseases: a review,'' \emph{Journal of
  the Royal Society Interface}, vol.~7, no.~50, pp. 1247--1256, 2010.

\bibitem{weitz2020awareness}
J.~S. Weitz, S.~W. Park, C.~Eksin, and J.~Dushoff, ``Awareness-driven behavior
  changes can shift the shape of epidemics away from peaks and toward plateaus,
  shoulders, and oscillations,'' \emph{Proceedings of the National Academy of
  Sciences}, vol. 117, no.~51, pp. 32\,764--32\,771, 2020.

\bibitem{weitz2016oscillating}
J.~S. Weitz, C.~Eksin, K.~Paarporn, S.~P. Brown, and W.~C. Ratcliff, ``An
  oscillating tragedy of the commons in replicator dynamics with
  game-environment feedback,'' \emph{Proceedings of the National Academy of
  Sciences}, vol. 113, no.~47, pp. E7518--E7525, 2016.

\bibitem{satapathi2023coupled}
A.~Satapathi, N.~K. Dhar, A.~R. Hota, and V.~Srivastava, ``Coupled evolutionary
  behavioral and disease dynamics under reinfection risk,'' \emph{IEEE
  Transactions on Control of Network Systems}, 2023.

\bibitem{khazaei2021disease}
H.~Khazaei, K.~Paarporn, A.~Garcia, and C.~Eksin, ``Disease spread coupled with
  evolutionary social distancing dynamics can lead to growing oscillations,''
  in \emph{2021 60th IEEE Conference on Decision and Control (CDC)}.\hskip 1em
  plus 0.5em minus 0.4em\relax IEEE, 2021, pp. 4280--4286.

\bibitem{frieswijk2022modeling}
K.~Frieswijk, L.~Zino, A.~S. Morse, and M.~Cao, ``Modeling the co-evolution of
  climate impact and population behavior: A mean-field analysis,'' \emph{arXiv
  preprint arXiv:2211.11075}, 2022.

\bibitem{tilman2020evolutionary}
A.~R. Tilman, J.~B. Plotkin, and E.~Ak{\c{c}}ay, ``Evolutionary games with
  environmental feedbacks,'' \emph{Nature communications}, vol.~11, no.~1, p.
  915, 2020.

\bibitem{gong2022limit}
L.~Gong, W.~Yao, J.~Gao, and M.~Cao, ``Limit cycles analysis and control of
  evolutionary game dynamics with environmental feedback,'' \emph{Automatica},
  vol. 145, p. 110536, 2022.

\bibitem{stella2022lower}
L.~Stella, W.~Baar, and D.~Bauso, ``Lower network degrees promote cooperation
  in the prisoner's dilemma with environmental feedback,'' \emph{IEEE Control
  Systems Letters}, vol.~6, pp. 2725--2730, 2022.

\bibitem{paarporn2023sis}
K.~Paarporn and C.~Eksin, ``Sis epidemics coupled with evolutionary social
  distancing dynamics,'' in \emph{2023 American Control Conference (ACC)},
  2023, pp. 4308--4313.

\bibitem{stella2021impact}
L.~Stella and D.~Bauso, ``The impact of irrational behaviors in the optional
  prisoner's dilemma with game-environment feedback,'' \emph{International
  Journal of Robust and Nonlinear Control}, 2021.

\bibitem{arefin2021imitation}
M.~R. Arefin and J.~Tanimoto, ``Imitation and aspiration dynamics bring
  different evolutionary outcomes in feedback-evolving games,''
  \emph{Proceedings of the Royal Society A}, vol. 477, no. 2251, p. 20210240,
  2021.

\bibitem{sandholm2010population}
W.~H. Sandholm, \emph{Population games and evolutionary dynamics}.\hskip 1em
  plus 0.5em minus 0.4em\relax MIT Press, 2010.

\bibitem{blume2003noise}
L.~E. Blume, ``How noise matters,'' \emph{Games and Economic Behavior},
  vol.~44, no.~2, pp. 251--271, 2003.

\bibitem{marden2012revisiting}
J.~R. Marden and J.~S. Shamma, ``Revisiting log-linear learning: Asynchrony,
  completeness and payoff-based implementation,'' \emph{Games and Economic
  Behavior}, vol.~75, no.~2, pp. 788--808, 2012.

\bibitem{tatarenko2014proving}
T.~Tatarenko, ``Proving convergence of log-linear learning in potential
  games,'' in \emph{2014 American Control Conference}.\hskip 1em plus 0.5em
  minus 0.4em\relax IEEE, 2014, pp. 972--977.

\bibitem{blume1995statistical}
L.~E. Blume, ``The statistical mechanics of best-response strategy revision,''
  \emph{Games and economic behavior}, vol.~11, no.~2, pp. 111--145, 1995.

\bibitem{auletta2012metastability}
V.~Auletta, D.~Ferraioli, F.~Pasquale, and G.~Persiano, ``Metastability of
  logit dynamics for coordination games,'' in \emph{Proceedings of the
  twenty-third annual ACM-SIAM symposium on Discrete algorithms}.\hskip 1em
  plus 0.5em minus 0.4em\relax SIAM, 2012, pp. 1006--1024.

\bibitem{paarporn2020impact}
K.~Paarporn, B.~Canty, P.~N. Brown, M.~Alizadeh, and J.~R. Marden, ``The impact
  of complex and informed adversarial behavior in graphical coordination
  games,'' \emph{IEEE Transactions on Control of Network Systems}, vol.~8,
  no.~1, pp. 200--211, 2020.

\bibitem{paarporn2020risk}
K.~Paarporn, M.~Alizadeh, and J.~R. Marden, ``A risk-security tradeoff in
  graphical coordination games,'' \emph{IEEE Transactions on Automatic
  Control}, vol.~66, no.~5, pp. 1973--1985, 2020.

\bibitem{zhang2023rationality}
Y.~Zhang and M.~M. Vasconcelos, ``Rationality and connectivity in stochastic
  learning for networked coordination games,'' 2023.

\bibitem{katz2014introduction}
J.~Katz and Y.~Lindell, ``Introduction to modern cryptography,'' 2014.

\bibitem{guckenheimer2013nonlinear}
J.~Guckenheimer and P.~Holmes, \emph{Nonlinear oscillations, dynamical systems,
  and bifurcations of vector fields}.\hskip 1em plus 0.5em minus 0.4em\relax
  Springer Science \& Business Media, 2013, vol.~42.

\end{thebibliography}

\appendices

\appendix


Item 3 of Theorem \ref{thm:logit_FPs} relies on the following technical result.

\begin{lemma}\label{lem:Tbeta}
	The function $T_\beta (x) := \beta^\inv \log(\frac{x}{1-x}): (0,1) \rightarrow \mbb{R}$ possesses the following properties.
	\begin{enumerate}
		\item It is continuous and strictly increasing on $x \in (0,1)$.
		\item $T_\beta(x) < 0$ for $x \in (0,{\half})$, $T_\beta({\half}) = 0$, and $T_\beta(x) > 0$ for $x \in ({\half},1)$.
		\item $\lim_{x \rightarrow 0_+} T_\beta(x) = -\infty$ and $\lim_{x \rightarrow 1_-} T_\beta(x) = +\infty$.
		\item For any fixed $x \in (0,\half)$, $T_\beta(x)$ is strictly increasing in $\beta$, $\limb T_\beta(x) = 0$, and $\lim_{\beta \rightarrow 0} T_\beta(x) = -\infty$.
		\item For any fixed $x \in (\half,1)$, $T_\beta(x)$ is strictly decreasing in $\beta$, $\limb T_\beta(x) = 0$ and $\lim_{\beta \rightarrow 0} T_\beta(x) = +\infty$.
		\item $T_\beta(x)$ is strictly concave on $x\in(0,\half)$ and strictly convex on $x\in(\half,1)$.
	\end{enumerate}
\end{lemma}

\begin{proof}
	We omit the proof of some items for brevity. 
	
	\noindent 1) $T_\beta (x)$ is continuous since it is a composition of continuous functions. It is increasing since $T'_\beta(x) = \frac{1}{\beta x(1-x)} > 0$ for all $x \in (0,1)$.

	\noindent 4) For any $x \in (0,\half)$,  $T_\beta(x) < 0$ and it strictly increases to 0 as $\beta \rightarrow \infty$.
	
	\noindent 5) similar argument to 4).
	
	\noindent 6) The second derivative is $T''_\beta(x) = -\frac{1-2x}{\beta x^2(1-x)^2}$, which is negative for $x \in (0,\half)$ and positive for $x \in (\half,1)$.
\end{proof}

We are now ready to provide the proof of Theorem \ref{thm:logit_FPs}.

\begin{figure}
            	\centering
            	\includegraphics[scale=0.33]{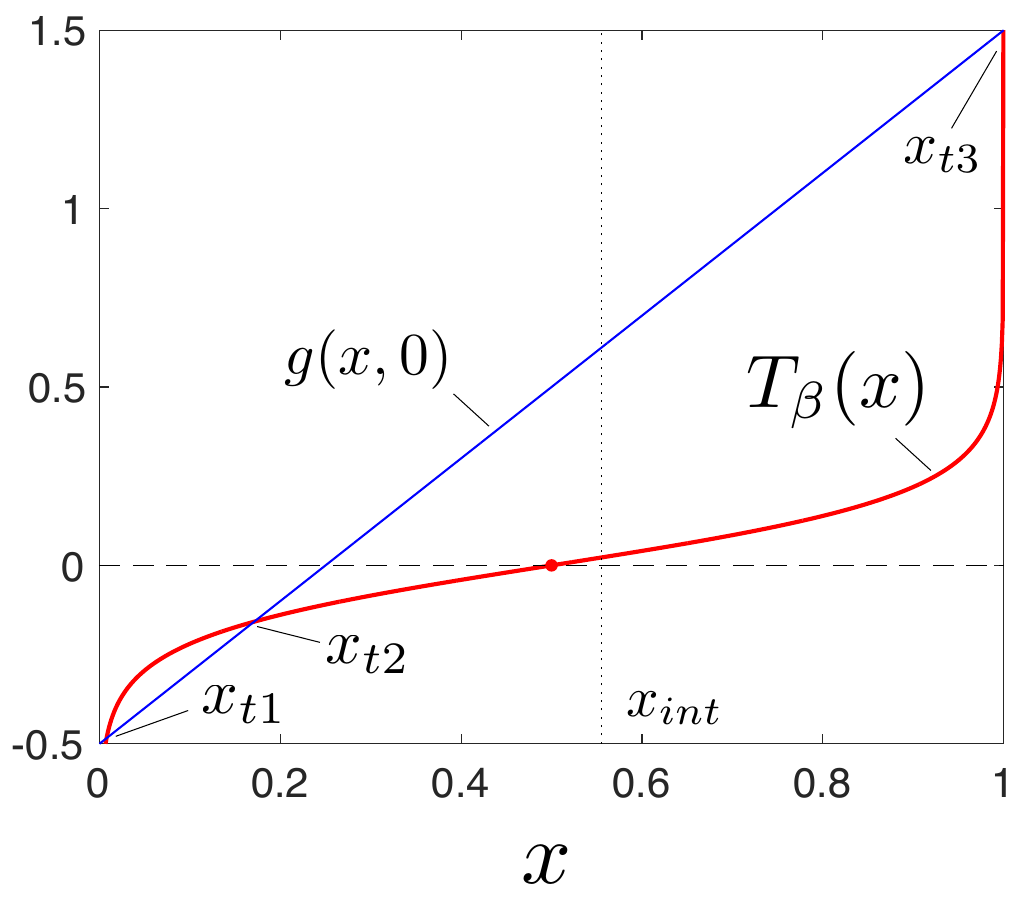}
            	\caption{The TOC fixed points are determined by the solutions of the equation $g(x,0) = bx + d = T_\beta(x)$. }
            	\label{fig:Tbeta}
\end{figure}

\begin{proof}[Proof of Theorem \ref{thm:logit_FPs}]
	Any fixed point $(x^*,n^*) \in \Gamma$ of system \eqref{eq:logit_system} must satisfy $F_1(x,n) = 0$ and $F_2(x,n) = 0$.
	
	\vspace{1mm}
	
		\noindent 1) Suppose $\beta = 0$. Then $F_1(x^*,n^*) = {\half} - x^* = 0$, and so any fixed point must have $x^* = 1/2$. If $\theta \neq 1$, the $F_2({\half},n^*) = n^*(1-n^*)(\frac{1+\theta}{2} - 1) = 0$ if and only if $n^* = 0$ or $n^* = 1$. In the case that $\theta = 1$, then  $F_2({\half},n^*) = 0$ for any $n^* \in [0,1]$. 
		
		\vspace{1mm}
		
		For the rest of the proof, we suppose $\beta > 0$.

		\noindent 2) From $F_2(x^*,n^*) = n^*(1-n^*)((1+\theta)x^* - 1) = 0$, any interior fixed point must have $x^* = x_\int = (1+\theta)^\inv$. Then, the equilibrium environmental level must satisfy the equation
		\begin{equation}
			\begin{aligned}
				\frac{ e^{\beta g((1+\theta)^\inv,n)} }{1 + e^{\beta g((1+\theta)^\inv,n)}} = (1+\theta)^\inv \\
				\Rightarrow g((1+\theta)^\inv,n) = \beta^\inv \log \theta^\inv.
			\end{aligned}
		\end{equation}
		The LHS above is a linear function in $n$, and it can be solved to obtain
		\begin{equation}
			n_\int = \frac{\dRT0 + \theta \dSP0 - \beta^\inv(1+\theta) \log \frac{1}{\theta} }{\dRT0 + \theta\dSP0 + \dTR1 + \theta\dPS1} 
		\end{equation}
		One can verify that the value $n_\int$ is feasible (lies in $(0,1)$) if and only if $\beta > \beta_\int = \frac{(1+\theta)\log \theta^\inv}{\dRT0 + \theta\dSP0}$.
		
		\vspace{1mm}
		
		\noindent 3) A TOC fixed point requires $n=0$, from which we immediately get $F_2(x,0) = 0$. To find $x$, we must solve the equation
		\begin{equation}\label{eq:TOC_equation}
			g(x,0) = bx + d  = \beta^\inv \log \frac{x}{1-x} = T_\beta(x).
		\end{equation}
		A reference plot is shown in Figure \ref{fig:Tbeta}. From Assumption \ref{assume:A0}, $b > 0$, $d < 0$, and $b+d > 0$. Also, the value $x_0 = -d/b = \frac{|\dSP0|}{\dRT0 + |\dSP0|} < 1/2$. From this and from Lemma \ref{lem:Tbeta}, we can deduce the following:
		
		For any $\beta>0$, \eqref{eq:TOC_equation} always has exactly one solution in the interval $(\half,1)$. To see this, suppose there are no solutions in this interval. We have that $g(x,0)$ takes values in the positive range  $(b/2 + d, b+d)$. However, for any $v > 0$, there must exist $x \in (\half,1)$ for which $T_\beta(x) = v$. This follows from Lemma \ref{lem:Tbeta} (item 1, 2, and 3). So, there must be at least one solution.

		Now, suppose there are two solutions in $(\half,1)$, $x_1 < x_2$. Note that there can be at most two solutions between a linear and convex function. Since we have that $g(\half,0) > T_\beta(\half)$, the first intersection point $x_1$ necessarily satisfies $\frac{\prt T_\beta}{\prt x}(x_1) > b$. In other words, the slope of $T_\beta$ must be greater than the slope of $g$ at $x_1$. But then there cannot exist a second solution $x_2$ since $\frac{\prt T_\beta}{\prt x}$ is increasing.
		
		Let us refer to the solution in the interval $(\half,1)$ as $x_{\toc3}(\beta)$. As a function of $\beta$, it is strictly increasing. This follows from item 5 in Lemma \ref{lem:Tbeta}. Moreover, it follows that $\limb x_{\toc3}(\beta) = 1$. 
		
		Now, we focus on the interval $x\in(0,\half)$ where . $T_\beta'(x) = \frac{1}{\beta x(1-x)}$. Observe that since $T_\beta(x)$ is strictly concave and $g(\half,0) > T_\beta(\half)$, \eqref{eq:TOC_equation} will either have no solutions on $(0,\half)$ or two solutions $x_1 \leq x_2 \in (0,\half)$ with equality if and only if $T_\beta'(x_1) = b$. In the latter case, it necessarily holds that $T_\beta'(x_1) > b$ and $T_\beta'(x_2) < b$. We can derive the value $x_b \in (0,\half)$ for which $T_\beta'(x_b) = b$ to be
		\begin{equation}
			x_b = \half - \half\sqrt{1 - \frac{4}{\beta b}}.
		\end{equation}
		We see that no such value can exist unless $\beta > 4/b$. Therefore, there exists a $\hat{\beta} \geq 4/b$ for which the two solutions $x_1 = x_2$. For any $\beta > \hat{\beta}$, the solution $x_1$ is decreasing in $\beta$ towards 0, and the solution $x_2$ is increasing in $\beta$ towards $x_0$ (due to Lemma \ref{lem:Tbeta} item 4).
		
		\vspace{1mm}
		
		\noindent 4) We have that for $n=1$, $F_2(x,1) = 0$. To find $x$, we need to solve
		\begin{equation}\label{eq:prosp_equation}
			g(x,1) = (a+b)x + c+d  = \beta^\inv \log \frac{x}{1-x} = T_\beta(x).
		\end{equation}
		By Assumption \ref{assume:A0},  $g(0,1) = c+d = -\dPS1 < 0$ and $g(1,1) = a+b+c+d = -\dTR1 < 0$. Then $g(x,1) < 0$ for all $x \in (0,1)$. It follows that any solution of \eqref{eq:prosp_equation} must lie in $(0,\half)$. If $a+b < 0$, then there is exactly one such solution because $g(x,1)$ is decreasing while $T_\beta$ is increasing. If $a+b > 0$, then there is still exactly one such solution. This is because $g(\half,1) < T_\beta(0) = 0$, which makes it impossible for the linear function $g(x,1)$ to intersect $T_\beta(x)$ at more than one point. Lastly, the unique solution $x^*\in (0,\half)$ must be decreasing in $\beta$ to 0 (due to Lemma \ref{lem:Tbeta} item 4).
\end{proof}

\end{document}